\documentclass[smallextended]{svjour3-arxiv}     

\usepackage{amsmath}

\usepackage{amsthm}

\smartqed  
\usepackage{graphicx}
\usepackage{wrapfig}
\usepackage{url}
\newcommand\eat[1]{}
\usepackage{tikz}
\usepackage{booktabs}  
\usetikzlibrary{arrows}
\usepackage{tablefootnote}
\usepackage{slashbox}
	
        \journalname{Working Paper}

\usepackage{array,xspace,multirow,hhline,graphicx,xcolor,tikz,colortbl,tabularx,amsmath,amssymb,amsfonts}
\usepackage[round]{natbib}
\usetikzlibrary{shapes}
\usepackage{algorithm,algorithmic}
\usepackage{mathrsfs}
\usepackage{enumitem}
\setenumerate[1]{label=(\emph{\roman{*}}),ref=(\emph{\roman{*}}),leftmargin=*}

\usepackage{txfonts}
\usepackage{eucal} 

\newcommand{\pref}{\succsim\xspace}

\usepackage{tikz}

\newcommand{\midd}{\mathbin{: }}

\usepackage{mathtools}
\DeclarePairedDelimiter\ceil{\lceil}{\rceil}

	\newcommand{\Pref}[1][]{
		\ifthenelse{\equal{#1}{}}{\mathrel \succsim}{\mathop{R_{#1}}}
	}    
				\newcommand{\spref}{\ensuremath{\succ}}                                      
	\newcommand{\sPref}[1][]{                  
		\ifthenelse{\equal{#1}{}}{\mathrel \succ}{\mathop{P_{#1}}}
	}                                          
	\newcommand{\Indiff}[1][]{                 
		\ifthenelse{\equal{#1}{}}{\mathrel \sim}{\mathop{\sim_{#1}}}
	}
	\newcommand{\prefset}[1][]{\ifthenelse{\equal{#1}{}}{\mathcal{R}}{\mathcal{R}_{#1}}}
	
	\newcommand{\indiff}{\mathbin \sim\xspace}

\let\enumtemp=\enumerate
\def\enumerate{\enumtemp\itemsep 1pt}
\let\itemtemp=\itemize
\def\itemize{\itemtemp\itemsep 1pt}

\newcommand{\Omit}[1]{}

            \newcommand{\B}{\mathcal{B}}
             \newcommand{\W}{\mathcal{W}}

	\def\ej{\color{black}}
	
	\newcommand{\lang}[1]{}

\begin{document}

	\title{Computing and Testing Pareto Optimal Committees}

	\author{Haris Aziz \and J{\'{e}}r{\^{o}}me Lang \and J{\'{e}}r{\^{o}}me Monnot}

	\institute{%
	  H. Aziz  \at
	 Data61, CSIRO and UNSW,
	 	  Sydney 2052 , Australia \\
	 	  Tel.: +61-2-8306\,0490 \\
	 	  Fax: +61-2-8306\,0405 \\
	 \email{haris.aziz@unsw.edu.au}\\\\
	  J. Lang and J. Monnot  \at
	   LAMSADE, Universit\'e Paris-Dauphine\\
	 Paris, France\\
	 \email{{\{lang, jerome.monnot\}@lamsade.dauphine.fr}}
}

%
%

	\newlength{\wordlength}
	\newcommand{\wordbox}[3][c]{\settowidth{\wordlength}{#3}\makebox[\wordlength][#1]{#2}}
	\newcommand{\mathwordbox}[3][c]{\settowidth{\wordlength}{$#3$}\makebox[\wordlength][#1]{$#2$}}
    		\renewcommand{\algorithmicrequire}{\wordbox[l]{\textbf{Input}:}{\textbf{Output}:}} 
    		 \renewcommand{\algorithmicensure}{\wordbox[l]{\textbf{Output}:}{\textbf{Output}:}}

\date{}

	%

\maketitle

\begin{abstract}
Selecting a set of alternatives based on the preferences of agents is  an important problem in committee selection and beyond. Among the various  criteria put forth for desirability of a committee, Pareto optimality is a minimal and important requirement.
As asking agents to specify their preferences over exponentially many subsets of alternatives is practically infeasible,
we assume that each agent specifies a weak order on single alternatives, from which a preference relation over subsets is derived using some preference extension.
We consider five prominent 
extensions (responsive,  downward lexicographic, upward lexicographic, best, and worst). For each of them, we consider the corresponding Pareto optimality notion, and we
study the complexity of computing and verifying Pareto optimal outcomes. We also consider strategic issues: for four of the set extensions, we present  a linear-time, Pareto optimal and strategyproof algorithm that even works for weak preferences.
\end{abstract}



\section{Introduction}

Pareto optimality is a central concept in economics and has been termed the \emph{``single most important tool of normative economic analysis''}~\citep{Moul03a}.
An outcome is Pareto optimal if there does not exist another outcome that all agents like at least as much and at least one agent strictly prefers.
Although Pareto optimality has been considered extensively in single-winner voting and other social choice settings such as fair division or hedonic games, it has received only little attention in  \emph{multiwinner voting}, in which the outcomes are \emph{sets} of alternatives.
Multiwinner voting applies to selecting a set of plans or a committee, hiring team members, movie recommendations, and more.  For convenience, we use the terminology ``committee'' even if our results have an impact far beyond committee elections~\citep{FSST17a,ABES17a}.

In \emph{single-winner voting} setting, agents express preferences over alternatives and a single alternative is selected. Pareto optimality in this context is straightforward to define, achieve, and verify. In \emph{multiwinner voting}, a well-known difficulty is that it is 
unrealistic to assume that 
agents will report 
preferences over all possible committees, since there is an exponential number of them. For this reason, most approaches assume that they only report a small part of their preferences, and that some {\em extension principle} is used to induce a preference over all possible subsets from this `small input' over single alternatives~\citep{BBP04a}.  Such preference extensions are also widely used in other social choice settings such as fair division or matching. The most two widely used choices of `small inputs' in multiwinner voting are {\em rankings (linear orders) over alternatives} and {\em sets of approved alternatives}.
In this paper we make a choice that generalizes both of them: 
agents report {\em weak orders over single alternatives}. Then we consider five prominent preference extension principles: the \emph{responsive} extension, where a set of alternatives $S$ is at least as preferred as a set of alternatives $T$ if $S$ is obtained from $T$ by repeated replacements of an alternative by another alternative which is at least as preferred;
the {\em optimistic},  or  {\em `best'} (respectively {\em pessimistic}, or {\em `worst'}) extension, which orders subsets of alternatives according to their most (respectively, least) preferred element;
the {\em downward lexicographic} extension, 
a lexicographic refinement of the optimistic extension, and the {\em upward lexicographic} extension, 
a lexicographic refinement of the pessimistic (worst) extension.

The responsive extension~\citep{BBP04a,RoSo90a} can be seen as the ordinal counterpart of  additivity.
The downward lexicographic extension has been considered in various papers~\citep{Boss95a,LMX12a,KlamlerPR12}.
The `best' set extension has been considered in a number of approaches such as full proportional representation \citep{ChamberlinCourant83,Monr95a} and other committee voting settings~\cite{ELS15a}.
The `worst' set extension, also used by \citet{KlamlerPR12} and \citet{SFL15a}, captures settings where the impact of a bad alternative in the selection overwhelms the benefits of good alternatives: for instance, when the decision about a crucial issue will be made by {\em one} of the members of the committee but the agent ignores  which one and is risk-averse; or the case of a parent's preferences over a set of movies to be watched by a child. The `best' and `worst' set extensions have been used  in coalition formation~\citep{AzSa15a,Cech08a}. 


Although 
set extensions have been implicitly or explicitly considered in multiwinner voting,
 most of the computational work has  dealt with
\emph{specific} voting rules (see the related work section).
Instead,
we concentrate on
Pareto optimality, consider the computation and verification of Pareto optimal committees, as well as
the existence of a polynomial-time and strategyproof algorithm that returns Pareto optimal outcomes.

\paragraph{Contributions}

We consider Pareto optimality with respect to  the five aforementioned preference set extensions.
We present various connections between the Pareto optimality notions.
For each of the notions, we undertake a detailed study of complexity of computing and verifying Pareto optimal outcomes.
Table~\ref{table:summary-complexity} summarizes the complexity results.

An important take-home message of the results is that testing Pareto optimality or obtaining Pareto improvements over status-quo committees is computationally hard even though computing \emph{some} Pareto optimal committee is easy. For responsive and downward lexicographic extensions we give a \emph{complete} characterization of the complexity of testing Pareto optimality when preferences are dichotomous and the size of top equivalence class is two: unless P = NP, Pareto optimality can be tested in polynomial time if and only if the size of the first equivalence classes is at most two.
For the `best' extension, we show that even computing a Pareto optimal outcome is NP-hard. Another interesting contrast with the responsive set extension is that even when preferences are dichotomous and the size of top equivalence class is two, testing Pareto optimality is coNP-complete.
In contrast to the other extensions, for the `worst' extension, both problems of computing and verifying Pareto optimal outcomes admit polynomial-time algorithms.

We also consider the requirement of strategyproofness on top of Pareto optimality. 
We show that there exist linear-time  Pareto optimal
and strategyproof algorithms for committee voting even for weak preferences for four of the five set extensions. The algorithms can be considered as careful adaptations of serial dictatorship for 
committee voting.

			\begin{table}[t!]
                        \large
				\centering
				\scalebox{0.8}{
			\begin{tabular}{lll}
			\toprule
			&\textbf{Computation}&\textbf{Verification}\\
			\textbf{Set Extension}&&\\
			\midrule 
				\multirow{2}{*}{Responsive (RS)}&\multirow{2}{*}{in $\text{P}^{\text{IC}}$ (Th.~\ref{th:RS-ic})}&coNP-C~(Th.~\ref{theoParetoCommittee}), W[2]-hard\\
				&&in P---dich. prefs and  $tw\leq 2$ (Th.~\ref{theo2ParetoCommittee})\\
				\midrule
					Downward Lexicographic (DL)&in $\text{P}^{\text{IC}}$ (Th.~\ref{th:dl-compute}) &coNP-C~(Th.~\ref{cor:verify-dl}), W[2]-hard\\
					\midrule
						Upward Lexicographic (UL)&in $\text{P}^{\text{IC}}$ (Th.~\ref{th:ul-compute}) &coNP-C~(Th.~\ref{cor:verify-ul}), W[2]-hard\\
					\midrule
					\multirow{2}{*}{Best ($\mathcal{B}$)}&NP-hard~(Th.~\ref{th:comput-best})&\multirow{2}{*}{coNP-C, W[2]-hard~(Th.~\ref{theo3ParetoCommittee})}\\
					&in P for strict prefs&\\
					\midrule
			Worst ($\mathcal{W}$)&in  $\text{P}^{\text{IC}}$ (Th.~\ref{th:sp-worst}) &in P (Th.~\ref{th:verify-worst})\\
\bottomrule
			\end{tabular}
			}
						\caption{Complexity of computing and verifying Pareto optimal committees. $\text{P}^{\text{IC}}$  (coined by Christos Papadimitriou in a seminar at Simons Institute in 2015) indicates a class of problems in which agents provide the input and the problems admit  a strategyproof and polynomial-time algorithm.}

\label{table:summary-complexity}
			\end{table}

\section{Related Work}\label{related}

A first related stream of work involves studying {\em specific committee elections rules} from a computational  point of view (generally with little or no focus on Pareto optimality).
Our focus on determining whether a committee is Pareto optimal or on finding a Pareto optimal committee, is in some sense orthogonal to the study of committee election rules.
The simplest (and most widely used) rules for electing a committee, called {\em best-$k$} rules, compute a score
for each alternative based on the ranks, and the alternatives with the best $k$ scores are elected~\cite{EFSS14a,FSST16}.
Scoring-based extension principles have also been used by \citet{Darmann13}.  Note that the output of a best-$k$ rule is obviously Pareto-optimal for the preferences induced by this scoring function, but not necessarily with respect to other set extensions.

Klamler {\em et al.} \cite{KlamlerPR12} compute optimal committees under a weight constraint for a {\em single} agent (therefore optimality is equivalent to Pareto optimality), using several preference extensions including `worst', `best', and downward lexicographic.

The `best' ($\mathcal{B}$) extension principle has been used in a number of papers on committee elections by {\em full proportional representation}, starting with \citep{ChamberlinCourant83} and studied from a computational point of view in a long series of papers~
({\em e.g.}, \citep{PSZ08a,LuBo11d,BetzlerSlinkoUllmann13,SkowronFS15,ElkindI15}. These rules obviously output Pareto optimal committees {\em for $\mathcal{B}$}, but not necessarily for other extensions.

Some of the set extensions considered in this paper have corresponding analogues when extending preferences over alternatives to preferences over `lotteries over alternatives.' In particular, the RS set extension corresponds to \emph{SD (stochastic dominance)} lottery extension. Also the DL and UL set extensions considered in this paper correspond to DL and UL lottery extensions considered in works in probabilistic social choice~\citep{Bran13b,ABBH12a,Cho15a}.

Some works are based on the {\em Hamming extension}.
Each agent specifies his ideal committee and he prefers committees with less Hamming distance from the ideal committee.
The Hamming distance notion can be used to define 
specific rules such as {\em minimax approval voting}~\citep{BramsKilgourSanver07}, which selects the committee minimizing the maximum Hamming distance for the agents.
Although the output of minimax approval voting is not always Pareto-optimal for the Hamming extension,  there are good Pareto-optimal approximations of it \cite{CaragiannisKalaitzisMarkakis10}.
Note that for dichotomous preferences, the Hamming extension coincides with the responsive and the downward lexicographic extensions, therefore our computational results for responsive set extension for dichotomous preferences also hold for the Hamming and downward lexicographic extensions.


\smallskip
A second line of work concerns understanding the classes of rules that result in Pareto optimal outcomes.
Most works along this line bear on a different type of committee elections, called {\em designated-seat voting}, where candidates must declare the seat they contest  \citep{Benoit10}.\footnote{If there are exactly two candidates per seat, then designated voting is equivalent to {\em multiple referenda}, where a decision has to be taken on each of a series of yes-no issues.} Results about the existence or non-existence of 
Pareto optimal rules have been presented~\cite{SanverSanver06,Benoit10,CuhadaroluLaine12}.

	\section{Setup}
	We consider 
	a set of agents $N=\{1,\ldots, n\}$, a set of alternatives $A=\{a_1,\ldots, a_m\}$ and a preference profile $\pref=(\pref_1,\ldots,\pref_n)$ such that each $\pref_i$ is a complete and transitive relation over $A$.
		We write~$a \pref_i b$ to denote that agent~$i$ values 
		$a$ at least as much as 
		$b$ and use~$\spref_i$ for the strict part of~$\pref_i$, i.e.,~$a \spref_i b$ iff~$a \pref_i b$ but not~$b \pref_i a$. Finally, $\indiff_i$ denotes~$i$'s indifference relation, i.e., $a \indiff_i b$ iff both~$a \pref_i b$ and~$b \pref_i a$.

		The relation $\pref_i$ results in equivalence classes $E_i^1,E_i^2, \ldots, E_i^{k_i}$ for some $k_i$ such that $a\spref_i a'$ if $a\in E_i^l$ and $a'\in E_i^{l'}$ for some $l<l'$.
		We will use these equivalence classes to represent the preference relation of an agent as a preference list
			$i\midd E_i^1,E_i^2, \ldots, E_i^{k_i}$.
		For example, we will denote the preferences $a\indiff_i b\spref_i c$ by the list $i:\ \{a,b\}, \{c\}$.
        An agent $i$'s preferences are \emph{strict} if
        the size of each equivalence class is 1.
An agent $i$'s preferences are \emph{dichotomous} if he partitions the alternatives into just two equivalence classes, i.e., $k_i=2$.
{Let $Topwidth(\pref)$ be the maximum size of the most preferred equivalence class, i.e., $Topwidth(\pref)=\max_{i\leq n}|E_i^1|$.}
For any $S\subseteq A$, we will denote by $\max_{\pref_i}(S)$ and $\min_{\pref_i}(S)$ the alternatives in $S$ that are maximally and minimally preferred by $i$ respectively. Thus, if  $q$ and $r$
are respectively the smallest and the largest indices such that
$E_i^{q}\cap S\neq \emptyset$ and $E_i^{r}\cap S\neq \emptyset$, then $\max_{\pref_i}(S)=E_i^{q}\cap S$ and  $\min_{\pref_i}(S)=E_i^{r}\cap S$.	
For $k \leq m$, let $S_k(A) = \{ W \subseteq A\midd |W| = k\}$.
	
	\section{Set Extensions and Pareto Optimality}

	\paragraph{Set Extensions}

{\em Set extensions} are used 
for reasoning about the preferences of an agent over \emph{sets} of alternatives given their preferences over single alternatives. \ej
For fixed-size committee voting, the 
\emph{responsive extension}
is very natural and has been applied in various matching settings as well~\citep{BBP04a,RoSo90a}.
For all $V, W\in S_k(A)$, we say that $W \mathrel{\pref_i^{RS}} V$ if and only if  there is an injection $f$ from $V$ to $W$ such that for each {$a\in V$, agent $i$ weakly prefers $f(a)$ to $a$, i.e. $f(a)\pref_i a$.}

We define the `best' set extension and the `worst' set extension which are denoted $\B$ and $\W$ respectively.
For all $W,V\in S_k(A)$, $W \succsim_i^{\B} V$ if and only if
{	$w\pref_i v$ for $w\in \max_{\pref_i}(W)$ and $v\in \max_{\pref_i}(V)$.
    On the other side,
$W \succsim_i^{\W} V$ if and only if $w\pref_i v$ for $w\in \min{\pref_i}(W)$ and $v\in \min_{\pref_i}(V)$.}

	
In the   \emph{downward lexicographic (DL)} extension,
an agent prefers a committee that selects more alternatives from his most preferred equivalence class, in case of equality, the one with more alternatives for the second most preferred equivalence class, and so on.
Formally, $W \succ_i^{DL} V$ iff for the smallest (if any) $l$ with $|W\cap E^l_i| \neq |V\cap E^l_i|$ we have $|W\cap E^l_i| > |V\cap E^l_i|$.

In the   \emph{upward lexicographic (UL)} extension,
an agent prefers a committee that selects less alternatives from his least preferred equivalence class, in case of equality, the one with less alternatives for the second least preferred equivalence class, and so on.
Formally, $W \succ_i^{UL} V$ iff for the largest (if any) $l$ with $|W\cap E^l_i| \neq |V\cap E^l_i|$ we have $|W\cap E^l_i| < |V\cap E^l_i|$.

\begin{remark}
	Consider an agent $i$ with preferences $\pref_i$ over $A$. Let $S, T\subset W$ such that $|S|=|T|=k$.
	Then,
	\begin{itemize}
		\item $S\pref_i^{RS} T\implies S\pref_i^{DL} T \implies S\pref_i^{\B} T$
		\item $S\pref_i^{RS} T\implies S\pref_i^{UL} T \implies S\pref_i^{\W} T$
		\item $S\succ_i^{RS} T\implies S\succ_i^{DL} T $
		\item $S\succ_i^{RS} T\implies S\succ_i^{UL} T $
		
	\end{itemize}
	The relations  follow from the definitions. 
	\end{remark}


%
%

		\paragraph{Efficiency based on Set Extensions}
With each set extension $\mathcal{E}$, we can define Pareto optimality with respect to $\mathcal{E}$.
A committee $W \in S_k(A)$ is \emph{Pareto optimal} with respect to 
$\mathcal{E}$, or  simply {\em $\mathcal{E}$-efficient},
if there exists no committee $W'  \in S_k(A)$ such that $W' \pref_i^{\mathcal{E}} W$ for all $i\in N$ and  $W' \succ_i^{\mathcal{E}} W$ for some $i\in N$.
Note that for each of our set extensions, $\mathcal{E}$-efficiency coincides with standard Pareto optimality when $k=1$. An outcome is a Pareto improvement over another if each agent weakly improves and at least one agent strictly improves.

\begin{example}
	Consider the preference profile:
	\begin{align*}
		1: a, b, c, d\\
		2: d, c, b, a
		\end{align*}

		Suppose $k=2$. Then,
		
		\begin{itemize}
			\item  The unique $\B$-efficient committee is $\{a,d\}$ 
			\item The unique $\W$-efficient committee is $\{b,c\}$.
		\item The DL-efficient committees are $\{a,d\}$, $\{a,b\}$, and $\{d,c\}$.
		\item The  UL-efficient committees are  $\{b,c\}$, $\{a,b\}$, and $\{d,c\}$. 
		\item The RS-efficient committees are $\{a,d\}$, $\{b,c\}$, $\{a,b\}$, and $\{d,c\}$,
		\end{itemize}
		
	\end{example}

    \begin{remark}
	    Consider a committee $S$.
	    \begin{itemize}
	    \item If $S$ is $DL$-efficient, then $S$ is $RS$-efficient
	    \item If $S$ is $UL$-efficient, then $S$ is $RS$-efficient

	    \end{itemize}

The argument is as follows. 
Suppose $S$ is not $RS$-efficient, then there exists some other outcome $T$ such that $T\pref_i^{RS} S$ for all $i\in N$ and $T\succ_i^{RS} S$ for some $i\in N$. In that case $T\pref_i^{DL} S$ for all $i\in N$ and $T\succ_i^{DL} S$ for some $i\in N$. Also  $T\pref_i^{UL} S$ for all $i\in N$ and $T\succ_i^{UL} S$ for some $i\in N$. Hence $S$ is neither $DL$-efficient not $UL$-efficient.
    \end{remark}
    
       \begin{remark}
     There always exists a $\mathcal{B}$-efficient committee that is also DL-efficient: DL Pareto improvements over a $\B$-efficient does not harm any agent with respect to the $\B$ relation. 
    \end{remark}
 
        \begin{remark}
           There always exists a $\mathcal{W}$-efficient committee that is also UL-efficient: UL Pareto improvements over a $\W$-efficient does not harm any agent with respect to the $\W$ relation. 
        \end{remark}

            In Figure~\ref{fig:relations}, we illustrate the relations between the different efficiency notions.  Later on in the paper we will present an algorithm that returns a committee that is $UL$-efficient and $DL$-efficient, and hence $RS$-efficient.

	\begin{figure}[htbp]
					 \centering
					   			\begin{tikzpicture}[scale=0.6]
					   				\centering
				\node (b) at (0,4) {$\mathcal B$-efficiency};
				\node (w) at (10,4) {$\mathcal W$-efficiency};
				\node (rs) at (5,-4) {$RS$-efficiency};
				\node (dl) at (0,0) {$DL$-efficiency};
				\node (ul) at (10,0) {$UL$-efficiency};
				\draw[densely dashed] (dl) to (b);
				\draw[densely dashed] (ul) to (w);
				\draw[densely dashed] (b) to (rs);
				\draw[densely dashed] (w) to (rs);
				\draw[densely dashed] (dl) to (ul);
				\draw[->] (dl) to (rs);
				\draw[->] (ul) to (rs);

					   			\end{tikzpicture}		
\caption{Relations between the five notions of efficiency. An arrow from ${\mathcal E}_1$-efficiency to ${\mathcal E}_2$-efficiency means that ${\mathcal E}_1$-efficiency implies ${\mathcal E}_1$-efficiency; a dashed line means there always exists a committee that is both ${\mathcal E}_1$- and ${\mathcal E}_2$-efficient;
 absence of arrow or line means that the sets of ${\mathcal E}_1$- and ${\mathcal E}_2$-efficient committees can be disjoint.}
\label{fig:relations}
\end{figure}
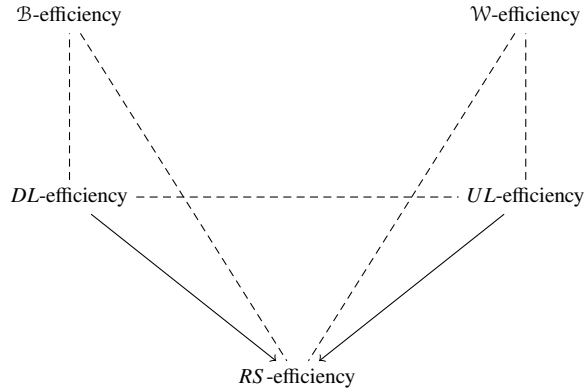

We also make the following general observation.


							\begin{lemma}
								If there is a polynomial-time algorithm to compute a Pareto improvement over a committee, then there exists a polynomial-time algorithm to {compute an $\mathcal{E}$-efficient committee} under set extensions $\mathcal{E}\in \{RS,DL,UL,  \mathcal{W}, \mathcal{B} \}$.
								\end{lemma}
								\begin{proof}
{Here, we start from any committee and we recursively apply Pareto improvement until we reach a Pareto optimal committee.}
								For the `best' and `worst' extensions, there can be at most $mn$ Pareto improvements because for one agent there can be at most $m$ improvements.
Since an $RS$-improvement implies an $DL$-improvement, let us bound the number of Pareto-improvements with respect to $DL$. 
In each Pareto-improvement, for the agent who strictly improves, 
the most preferred equivalence class that has different number of alternative in the outcome increases by at least one. Therefore the most preferred equivalence class can be the improving class in at most $m$ of the Pareto improvements. Similarly, the number of Pareto improvements in the subsequent less preferred equivalence class improves in a Pareto improvement can be at most $m$ of the Pareto improvements. Therefore the total number of DL Pareto-improvements is bounded by $m^2n$. A similar argument holds for UL as well. 
									\end{proof}



We end this section by observing 
that, under any of the set extensions we consider,  
a set of Pareto optimal alternatives may be Pareto dominated. Consider the following example.

\begin{example}\label{exa-pd}
					\begin{align*}
				1&: a,c,b,d
				&\quad2&: a,d,b,c\\
				3&: b,c,a,d
				&\quad4&: b,d,a,c
				\end{align*}
				The set $\{c,d\}$	consists of Pareto optimal alternatives but is Pareto dominated  by $\{a,b\}$ under any of our set extensions.	
\end{example}
				
				
	\section{Responsive Set Extension}

There  is a trivial way to achieve Pareto optimality under the responsive set extension by taking any {decreasing} scoring vector consistent with the ordinal preferences, finding the total score of each alternative and returning the set of $k$ alternatives with the maximum scores. For instance, on Example \ref{exa-pd}, the outcome of the rule that outputs the alternatives with the best $k$ Borda scores is $\{a,b\}$.

\begin{theorem}\label{th:rs-compute}
	A Pareto optimal committee under the responsive set extension committee can be computed in linear time.
	\end{theorem}

		In many situations, one may already have a status-quo committee and one may want to find a Pareto improvement over it. This problem of testing Pareto optimality  and finding a Pareto improvement under the responsive set extension turns out to be a much harder task.
Note that if there exists a polynomial-time algorithm to compute a
Pareto improvement, then it means that testing
Pareto optimality  is also polynomial-time solvable.

						\begin{theorem}\label{theoParetoCommittee}
Checking whether a committee is Pareto optimal under the responsive set extension is coNP-complete, even for dichotomous preferences and  $Topwidth(\pref) \geq 3$,  or for strict preferences.
						\end{theorem}

						\begin{proof}
We only present the case where $Topwidth(\pref)=3$. The reduction is from the NP-complete problem {\sc vertex cover} \cite{GaJo79a}.
Given a simple graph $G=(V,E)$, the {\sc minimum vertex cover} problem consists in finding a subset $C\subseteq V$ of minimum size such that every edge $e\in E$ is incident to some node of $C$. Its
decision version
{\sc vertex cover} takes as input
a simple graph $G=(V,E)$ and an integer $k$ and problem is deciding if there exists a vertex cover $C\subseteq V$ of $G$ with $|C|\leq k$.

Let $\langle (V,E), k \rangle$ be an instance of {\sc vertex cover}, with $[x,y]$ being one arbitrary edge in $E$.
We build  the following  instance of Pareto optimality under $RS$:
\begin{itemize}\setlength\itemsep{0em}
\item[$\bullet$] $N=\cup_{e\in E}N_e \cup \{a\}$, where for each edge $e\in E$, $N_e$ is a set of $k$ agents, and $a$ is a special agent.

						\item[$\bullet$] 
						$A=V\cup D$, where $D=\{d_1,\dots,d_k\}$.

						\item[$\bullet$]
						For each $e=[u,v]\in E$, the preferences of
						agent $e^i$, for $i=1,\dots, k$, and of agent $a$, are
						$$\begin{array}{ll} e^i:&~\{u,v,d_i\}, (D-d_i)\cup (V\setminus\{u,v\})\\
						a:& \{x,y\}, D\cup (V\setminus\{x,y\})
						\end{array} $$
						 \end{itemize}
						
					The reduction is clearly done within polynomial time and preferences are dichotomous.						
We can check easily that committee $D$ (of size $k$) is not Pareto optimal under 
$RS$ if and only if there exists a vertex cover of $G$ of size at most $k$.

		For strict preferences, in the previous reduction we replace $\{u,v,d_i\}, \ldots$ by  $\{u\},\{v\}, \{d_i\}, \ldots$ in the preferences of $e^i$.
		It is easy to see that the proof is similar.
						\end{proof}

Using a similar reduction from the {\sc Hitting Set} problem, we can also prove Theorem~\ref{theoParetoCommitteeW2} that concerns a parametrized complexity intractability result~\cite{DoFe13a}. {\sc Hitting Set} is defined as follows: given a ground set $X$ of elements, and a collection $\mathcal{C}=\{C_1,\ldots, C_{\ell}\}$ of subsets of $X$, does there exist a $H\subset X$ such that $|H|\leq k$ and  $H\cap C\neq \emptyset$ for all $C\in \mathcal{C}$?

				\begin{theorem} \label{theoParetoCommitteeW2}
				Checking whether a committee is  Pareto optimal under the responsive set extension is W[2]-complete under parameter $k$, even for dichotomous preferences.
				\end{theorem}

For dichotomous preferences we present a complete characterization of the complexity according to the $Topwidth(\pref)$ parameter.
If $Topwidth(\pref)= 1$, then in any Pareto improvement over committee $D$, any alternative in $D$ that is most preferred by some agent needs to be kept selected, and therefore the problem of checking $RS$-efficiency is easy. If $Topwidth(\pref)\geq 3$, from Theorem \ref{theoParetoCommittee}, the problem is hard. Remains the case $Topwidth(\pref)= 2$.

\begin{theorem}
							\label{theo2ParetoCommittee}
For dichotomous preferences, a Pareto improvement over a committee with respect to the responsive set extension can be computed in polynomial time
when $Topwidth(\pref)\leq 2$.
\end{theorem}
\begin{proof}
Consider
a preference profile $\pref=(\pref_1,\ldots,\pref_n)$ where each $\pref_i$ is dichotomous and verifies $Topwidth(\pref)= 2$,
and 
let $D \in S_k(A)$.
For each  $i\in N$, let $(E_i^1,E_i^2)$ be the partition associated with $\pref_i$.

First, if for all $i \in N$, $E_i^1 \subseteq D$, then $D$ is obviously $RS$-efficient. Assume it is not the case, that is,
(1) for some $i \in N$, $E_i^1 \setminus D \neq \emptyset$.
Let
\begin{itemize}
\item $N'=\{i\in N:E_i^1\cap D=E_i^1\}$, $W'=\cup_{i\in N'}E_i^1$ (by construction, $W'\subseteq D$), and $k'=|W'|$.
\item $N''=\{i\in N\setminus N':E_i^1\cap (D\setminus W')\neq \emptyset\}$ and \mbox{$A''=\cup_{i\in N''}E_i^1$}.
\end{itemize}

Now, we build a graph $G=(V,E)$ with $V=\{v_1,\dots,v_r\}$ isomorphic to $A''$, and $[v_p,v_q]\in E$ iff $E_i^1=\{a_p,a_q\}$ for some $i\in N''$: each edge of $G$ corresponds to the top two alternatives of some agent, provided one of them is in $D\setminus W'$.
Let $\tau(G)$ be the size of an optimal vertex cover of $G$.

We first claim that there is a  Pareto improvement over $D$ if and only if one of follows two conditions is satisfied:
\begin{itemize}\setlength\itemsep{0em}
\item[$(i)$] $\tau(G)<k-k'$, or
\item[$(ii)$] $\tau(G)=k-k'$, and there is an optimal vertex cover
of $G$ containing either  at least an element of $E_i^1$ for some  $i\notin N'\cup N''$, or two elements of $E_i^1$ for some $i \in N''$.
\end{itemize}

We first show that (i) and (ii) are sufficient. If (i) holds then take a committee corresponding to a minimum vertex cover of $G$, add to it the $k'$ alternatives of $W'$, and
add $(k - k')- \tau(G)$ alternatives, with at least one in $\cup_i(E_i^1 \setminus D)$; this is possible because of (1). If (ii) holds, then take a committee corresponding to a minimum vertex cover of $G$, and add to it the $k'$ alternatives of $W'$. In both cases, the obtained committee contains $E_i^1$ for all $i \in N'$, contains at least one element of $E_i^1$ for all $i \in N'$, and contains either two elements of $E_i^1$ for some $i \in N''$, or an element of $E_i^1$ for some $i \notin N \cup N''$. Therefore it is a Pareto-improvement over $D$.

Now, we show that (i) and (ii) are necessary. Let $W\in S_k(A)$ be a  Pareto improvement of $D$ containing a maximum number of alternatives from $D$. We have the following two properties: $W'\subseteq W$ and $W\setminus W'$ is a vertex cover of $G$. $W'\subseteq W$ holds, since otherwise there would be an $i\in N'$ such that $W' \succsim_i^{RS} W$ does not hold.
For similar reasons, $C'=(W\setminus W')\cap A''$ is a vertex cover of $G$. If
$|(W\setminus W')\cap A''|< \tau(G)$, then by adding to it any set of $D\setminus C'$ of size $k-k'-\tau(G)$ we obtain a set of size $k$ which constitutes a Pareto improvement of $D$ because now, $E_i^1\subseteq W$ for some $i\in N''$. If $|(W\setminus W')\cap A''|= \tau(G)$, then
$(W\setminus W')\cap A''=W\setminus W'$ and necessarily either $E_i^1\cap C\neq\emptyset$ for some $i\notin (N'\cup N'')$ or $E_i^1\subseteq C$ for some $i\in N''$.

It remains to be shown that (i) and (ii) can be checked in polynomial time.
(i) can be done in polynomial-time because $G$ is bipartite: indeed, by construction, $G$ is two-colorable with color sets $A''\cap D$ and $A''\setminus D$, and
by K{\"{o}}nig's theorem, for bipartite graphs, the problem of finding the minimum vertex cover is equivalent to computing a maximum matching, hence solvable in polynomial time. As for (ii), if $\tau(G)=k-k'$, we have to check whether for some optimal vertex cover $C$ of $G$, either (ii.1) $E_i^1\cap C\neq\emptyset$ holds for  some $i\notin (N'\cup N'')$, or (ii.2)  $E_i^1\subseteq C$ for some $i\in N''$.
In order to check (ii.1), for each $i\notin (N'\cup N'')$ such that there exists $x \in E_i^1\cap A''$, we transform $G$ into a new bipartite graph $G_{ \{x\} }$ where we add a new vertex $x'$ and an edge $[x,x']$.
In order to check (ii.2), for each $i \notin N''$, let $E_i^1 = \{x,y\}$; we transform $G$ into a new bipartite graph $G_{ \{x,y\} }$ where we add two new vertices $x'$ and $y'$, and two edges $[x,x']$ and $[y,y']$.
Finally, we test if $\tau(G)=\tau(G_{ \{x\} })$ or if $\tau(G)=\tau(G_{ \{x,y\} })$ for one of these graphs, because all optimal vertex covers of $G_{ \{x\} }$ (respectively $G_{ \{x,y\} }$) must contain $x$ (respectively $\{x,y\}$).
\end{proof}

\begin{example}\label{example:topwidth-algo}
We illustrate the algorithm in the proof of Theorem~\ref{theo2ParetoCommittee}.
Let $k=2$ and consider the dichotomous profile, where we specify only the top equivalence class of each agent:
										\begin{align*}
										1&: \{a,c\} & 2&: \{b,c\}
										& 3&: \{b,d\} \\
										4&: \{d,e\} & 5 &: \{e,f\}
										\end{align*}
 Let $D=\{a,b\}$. We have $N' = W'=\emptyset$, $k'=0$, 
 $D\setminus W'=\{a,b\}$, $N''=\{1,2,3\}$, and $A''=\{a,b,c,d\}$.
We construct the graph
$G=(V,E)$: $V=\{v_a,v_b,v_c,v_d\}$ and $E=\{\{v_a,v_c\}, \{v_b,v_c\}, \{v_b,v_d\}\}$.
We have $\tau(G)=2=k-k'$.
Now we consider the four graphs $G_{\{d\}}$, resulting from the addition to $G$ of a new vertex $v_{d'}$ and edge $[v_d,v_{d'}]$, and
$G_{\{a,c\}}$, $G_{\{b,c\}}$ and $G_{\{b,d\}}$: $G_{\{a,c\}}$ results from the addition to $G$ of two new vertices $v_{a'}, v_{c'}$ and edges $[v_{a},v_{a'}]$ and $[v_c,v_{c'}]$, etc.
Two of these graphs have an optimal cover of size 2: $G_{\{d\}}$, with optimal cover $\{v_c,v_d\}$, and $G_{\{b,c\}}$, with optimal cover $\{v_b,v_c\}$. Therefore, $\{c,d\}$ and $\{b,c\}$ are $RS$-Pareto-improvements over $\{a,b\}$, and $\{a,b\}$ is not $RS$-efficient.
									\vspace{-1mm}
				
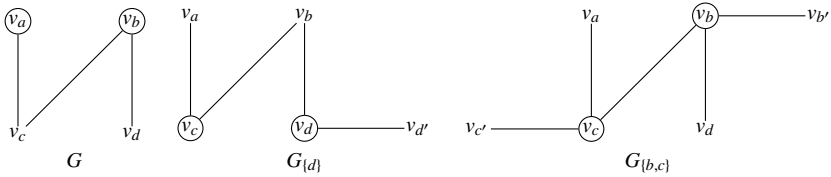
\begin{figure}[htbp]
\centering
\begin{tabular}{ccc}
\begin{tikzpicture}[scale=0.5]
\centering
\tikzstyle{pfeil}=[->]
\tikzstyle{onlytext}=[]
\tikzstyle{lab}=[font=\footnotesize\itshape]
\tikzstyle{agent}=[minimum size=1pt, inner sep=0.4pt]
\tikzstyle{cover}=[draw,circle,minimum size=1pt, inner sep=0.4pt]
\node[cover] (a) at (0,3) {$v_a$};
\node[cover] (b) at (3,3) {$v_b$};
\node[agent] (d) at (3,0) {$v_d$};
\node[agent] (c) at (0,0) {$v_c$};
\draw[-] (a) to [] (c);
\draw[-] (b) to [] (d);
\draw[-] (b) to [] (c);
\end{tikzpicture}
&
\begin{tikzpicture}[scale=0.5]
\centering
\tikzstyle{pfeil}=[->]
\tikzstyle{onlytext}=[]
\tikzstyle{lab}=[font=\footnotesize\itshape]
\tikzstyle{agent}=[minimum size=1pt, inner sep=0.4pt]
\tikzstyle{cover}=[draw,circle,minimum size=1pt, inner sep=0.4pt]
\node[agent] (a) at (0,3) {$v_a$};
\node[agent] (b) at (3,3) {$v_b$};
\node[cover] (d) at (3,0) {$v_d$};
\node[agent] (dp) at (6,0) {$v_{d'}$};
\node[cover] (c) at (0,0) {$v_c$};
\draw[-] (a) to [] (c);
\draw[-] (b) to [] (d);
\draw[-] (b) to [] (c);
\draw[-] (d) to [] (dp);
\end{tikzpicture}	
&
\begin{tikzpicture}[scale=0.5]
\centering
\tikzstyle{pfeil}=[->]
\tikzstyle{onlytext}=[]
\tikzstyle{lab}=[font=\footnotesize\itshape]
\tikzstyle{agent}=[minimum size=1pt, inner sep=0.4pt]
\tikzstyle{cover}=[draw,circle,minimum size=1pt, inner sep=0.4pt]
\node[agent] (cp) at (0,0) {$v_{c'}$};
\node[agent] (a) at (3,3) {$v_a$};
\node[cover] (b) at (6,3) {$v_{b}$};
\node[agent] (bp) at (9,3) {$v_{b'}$};
\node[agent] (d) at (6,0) {$v_d$};
\node[cover] (c) at (3,0) {$v_c$};
\draw[-] (a) to [] (c);
\draw[-] (b) to [] (d);
\draw[-] (b) to [] (c);
\draw[-] (c) to [] (cp);
\draw[-] (b) to [] (bp);
\end{tikzpicture}
\\
\small $G$ & \small $G_{\{d\}}$ & \small $G_{\{b,c\}}$												
\end{tabular}
\caption{Graphs corresponding to Example~\ref{example:topwidth-algo}}
\label{fig:fttcgraph}
\end{figure}
\end{example}
									
Note that finding an algorithm that computes a Pareto improvement over a committee can be used to  decide whether a  given a committee $D$ of size $k$, is Pareto optimal under the responsive set extension.

	\paragraph{Pareto optimality and Strategyproofness}

We now try to achieve both $RS$-efficiency 
and strategyproofness simultaneously.
A mechanism $f$ is \emph{strategyproof} if reporting truthful preferences is a dominant strategy with respect to the responsive set extension: $f(\pref) \pref_i^{RS} f(\pref_i',\pref_{-i})$ for all preference profiles $\pref$ and $(\pref_i',\pref_{-i})$.
Note that defining strategyproofness in this way with respect to the RS extension is stronger than defining it for any of the other four extensions considered in this paper. Nonetheless, we will present some positive results with respect to strategyproofness. 


A naive way of achieving $RS$-efficiency and Pareto optimality is to enumerate the list of possible winning sets and implement serial dictatorship over the possible outcomes as is done in voting~\citep{ABBH12a}. However, the number of possible outcomes is exponential and responsive preferences result in a partial order over the possible winning sets and not a complete and transitive order. 
This problem is solved by Algorithm~\ref{algo:comtRSD} which can be viewed as a computationally efficient serial dictatorship.

			\begin{algorithm}[h!]
\caption{Committee Voting Serial Dictatorship}
			  \label{algo:comtRSD}
			\renewcommand{\algorithmicrequire}{\wordbox[l]{\textbf{Input}:}{\textbf{Output}:}}
			 \renewcommand{\algorithmicensure}{\wordbox[l]{ \textbf{Output}:}{\textbf{Output}:}}
			\algsetup{linenodelimiter=\,}
				Input: $(N,A,\pref, k,\text{~permutation~} \pi \text{ of } N)$\lang{what is $v$?}\\
				Output:   $W \in S_k(A)$.\\  
			  \begin{algorithmic}[1]

		\STATE $L$ (last set to be refined) $\longleftarrow$ $A$
		\STATE  $r$ (number of alternatives yet to be fixed) $\longleftarrow$ $k$; $W$ $\longleftarrow$ $\emptyset$
		\STATE $i'$ (index of the permutation $\pi$) $\longleftarrow$ $1$
		\WHILE{$r \neq 0$ or $i'\neq n$}
	\STATE Agent $i= \pi(i')$ selects first $t$ equivalence classes such that $|\bigcup_{j=1}^t E_i^j\cap L|\geq r$ and $|\bigcup_{j=1}^{t-1} E_i^j\cap L|< r$.
	\STATE $W$ $\longleftarrow$ $W\cup (\bigcup_{j=1}^{t-1} E_i^j\cap L)$ (we say agent $i$ \emph{fixes} the alternatives in $\bigcup_{j=1}^{t-1} E_i^j$);
	\STATE $r$ $\longleftarrow$ $|\bigcup_{j=1}^t E_i^j\cap L|-|\bigcup_{j=1}^{t-1} E_i^j\cap L|$
	\STATE  $L$ $\longleftarrow$ $E_i^t$; $r_{i'}$ $\longleftarrow$ $r$
	\STATE Increment $i'$ by one
		\ENDWHILE
        \IF{$r>0$}
        \STATE pick any $r$ alternatives from $L$ and add them to $W$ 
        \ENDIF
				\RETURN $W$		
			 \end{algorithmic}
			\end{algorithm}

	\begin{theorem}\label{th:RS-ic}
		There exists a linear-time and strategyproof algorithm that returns a committee that is  Pareto optimal under the responsive set extension.
		\end{theorem}
		\begin{proof}

	Consider Algorithm~\ref{algo:comtRSD}.
We show that at each stage $i'$, agent $\pi(i')$, implicitly refines the set of feasible committees to the maximal set of most preferred outcomes from the set by providing additional constraints.
This is true for the base case $i'=1$.
Now assume it holds from $1$ to $i'$.
Note that $L$ contains all those alternatives that are strictly less preferred by agents in $\{\pi(1),\ldots, \pi(i')\}$ than the ones they respectively  \emph{fixed}. Moreover, each agent in $\{1,\ldots, \pi(i')\}$ is indifferent between the alternatives in $L$.
As for $\pi(i'+1)$, he fixes the best $|\bigcup_{j=1}^{t-1} E_{\pi(i'+1)}^j\cap L|$ alternatives in $L$ where $t$ is the value such that $|(\bigcup_{j=1}^t E_{\pi(i'+1)}^j)\cap L|\geq r_{i'}$ and $|\bigcup_{j=1}^{t-1} E_{\pi(i'+1)}^j\cap L|< r_{i'}$.
For $E_{\pi(i'+1)}^t$, the agent only requires that
 $r_{i'+1}=|(\bigcup_{j=1}^{t-1} E_{\pi(i'+1)}^j) \cap L|-|(\bigcup_{j=1}^{t-1} E_{\pi(i'+1)}^j) \cap L|$ alternatives are selected from his equivalence class $E_{\pi(i'+1)}^t$ which is ensured by the definition of the algorithm.
		It follows from the argument that the returned set is Pareto optimal under the responsive set extension.
	For strategyproofness, when an agent $\pi(i')$ turn comes, it only has a choice over \emph{fixing} the alternatives in $L$ and requiring $r_{i'}$ alternatives from his equivalence class $E_{\pi(i')}^t$. In this case the algorithm already chooses one of the best possible committees for the agent.
	 \end{proof}

	Note that for $k=1$, the algorithm is equivalent to serial dictatorship as formalized by \citet{ABB13b}.
	Note that a committee that is Pareto optimal under the responsive set extension may not be a result of serial dictatorship. This holds even for $k=1$ and the basic voting setting.

The problem with the serial dictatorship algorithm formalized is that it overly favours the agent that is the first in the permutation. One way to limit his power is to let him choose only $\ceil{k/n}$ alternatives.
We note that this attempt at having a fairer extension of serial dictatorship comes at an expense because strategyproofness is compromised. Consider the profile in which 1 has preferences $a, b, c$ and 2 has preferences $a, c, b$. For $k=2$, and permutation $12$, the outcome is $\{a,c\}$. But if agent 1 reports $b,a,c$, then the outcome is $\{a,b\}$.

\section{`Best' Set Extension}

Next, we consider Pareto optimality with respect to ${\cal B}$,
 which has been used for defining many rules (see Section \ref{related}).

\begin{theorem}
							\label{theo3ParetoCommittee}
Unless P$=$NP, there is no polynomial-time algorithm to compute a Pareto improvement over a committee with respect to ${\cal B}$, even for dichotomous preferences and $Topwidth(\pref) = 2$.
\end{theorem}
\begin{proof}
We show that if it is not the case, then we can solve polynomially the vertex cover decision problem. Consider an instance of {\sc vertex cover}
given by a simple graph $G=(V,E)$ with $V=\{v_1,\dots,v_q\}$ and $E=\{e_1,\dots,e_r\}$,
and an integer $k$.
Assume the existence of a polynomial-time algorithm $\texttt{Algo}$  that computes a Pareto improvement over a committee with respect to ${\cal B}$ when \mbox{$Topwidth(\pref)\leq 2$}: given a profile $\pref$ and a set of $k$ alternatives $W$, \texttt{Algo}$(\pref,W)$ returns, in time polynomial in $|\pref|$, \texttt{Yes} if $W$ is Pareto optimal with respect to ${\cal B}$, and otherwise returns a $k$-set $U$ of alternatives which Pareto dominates $W$. We will now prove by applying at most $n$ times  \texttt{Algo} with different inputs that we can decide in polynomial if $G$ has a vertex cover $C\subseteq V$ of $G$ with $|C|\leq k$.
We construct the following profile $P$:

						 \begin{itemize}\setlength\itemsep{0em}
						\item[$\bullet$] The set of agents is $N=\{1,\dots,q+r-k\}$, where agent $i\leq r$  corresponds to edge $e_i\in E$.

						\item[$\bullet$] The set of $2q-k$ alternatives is 
						$A=V\cup D$ where $D=\{d_1,\dots,d_{q-k}\}$.

						\item[$\bullet$] 
						Let $e_i=[u,v]\in E$ be an edge of $G$; the preferences of
						agent $i$ for $i=1,\dots, r$  are: $$ i:~\{u,v\}, D\cup (V\setminus\{u,v\}).$$

						\noindent The preferences of the last set of  $q-k$ agents $\{r+1,\dots,q+r-k\}$ are given by: for
$i=1,\dots, q-k,$ $$r+i:~d_i, V\cup (D\setminus\{d_i\}).$$
						 \end{itemize}

						The reduction is clearly done within polynomial time and the set of preferences given by $\pref$ are dichotomous. \lang{Removed \begin{quote} We now focus of set of alternatives of size $n$.\end{quote} }

						\smallskip
Consider the following inductive procedure:
$W_0=V$ and for $i\geq 1$, $W_i= $ \texttt{Algo} $ (\pref, W_{i-1})$ if $W_{i-1}$ is not Pareto optimal with respect to ${\cal B}$, otherwise we return $W_{i-1}$. Let $W=W_{q-k}$ be the solution output after $q-k$ calls. Because  \texttt{Algo} is polynomial, the whole
procedure is polynomial.

 We claim that $G$ has a vertex cover of size $k$ iff $D\subseteq W$.
 We will first prove by induction that at each step $i$, $W_i\setminus D$ is a vertex cover of $G$. For the initial step, it is valid because $V$ is a vertex cover of $G$. Assume that it is true for $i<q-k$ and let us prove that $W_{i+1}\setminus D$ is a vertex cover of $G$. If it is not the case, some edge $e_j=[u,v]\in E$ is not covered. By assumption, $e_j$ is covered by $W_i\setminus D$. This implies
 \mbox{$W_{i} \succ_j^{\B} W_{i+1}$}, which is a contradiction. Hence, $W_{i+1} \succsim_j^{\B} W_{i}$.
 From this hypothesis, we deduce $D\setminus W_{i+1} \succsim_j^{\B} D\setminus W_{i}$ for $j=n+1,\dots,2n-k$ with a strict preference for some agent. Equivalently, $D\setminus W_{i} \subset D\setminus W_{i+1}$. In conclusion, after $q-k$ recursive calls, $|W\setminus D|\leq k$ if and only if $D\subseteq W$.
\end{proof}

\begin{theorem}
							\label{th:comput-best}
Computing a ${\cal B}$-efficient committee is NP-hard, even for dichotomous preferences.
\end{theorem}
		 	\begin{proof}
We give a 
reduction from {\sc Hitting Set}. Let
 $N=\{1,\dots, \ell\}$, $A=X$ and for each $i\in N$, $i$'s dichotomous  preferences are $i: C_i, (X\setminus C_i).$
If there exists a polynomial-time algorithm to compute a  ${\cal B}$-efficient
committee, it will return a committee in which each agent gets a most preferred alternative if such a committee exists.
			But such a committee corresponds to a hitting set of size $k$.
		\end{proof}

\section{Downward Lexicographic Set Extension}

We point out that for dichotomous preferences, the responsive set extension coincides with the 
downward lexicographic  set extension.
Hence we get a corollary of our results for responsive preferences:

\begin{corollary}\label{cor:verify-dl}
						Checking whether a committee is  $DL$-efficient
						 is coNP-complete, even for dichotomous preferences and
		$Topwidth(\pref) \geq 3$.
\end{corollary}

		Note that Algorithm~\ref{algo:comtRSD} returns a  $DL$-efficient  committee. The reason is that each agent in her turn refines the set of possible outcomes to her most preferred subset of outcomes. Each committee is the refined set is at least as preferred with respect to RS (and hence with respect to DL) to all committees in the set of possible outcomes. 

	\begin{theorem}\label{th:dl-compute}
		There exists a linear-time and strategyproof algorithm that returns a  $DL$-efficient committee.
		\end{theorem}

\section{`Worst' Set Extension}

In contrast to all the other set extensions considered in the paper, Pareto optimality with respect to the `worst' set extension can be checked in polynomial time.

						\begin{theorem}
													\label{th:verify-worst} There exists a polynomial-time algorithm that checks whether a committee is  ${\cal W}$-efficient
and computes a Pareto improvement over it if possible.
													\end{theorem}
								 	\begin{proof}
Let $W \in S_k(A)$.
For each $i\in N$, let $E_i^{t_i}$ be the least preferred equivalence class such that $E_i^{t_i}\cap W\neq \emptyset$.
We want to check whether there is a $k$-set $D$ of alternatives in which
at least some agent $i\in N$ gets a strictly better outcome and all the other agents get at least as preferred an outcome.
We check this as follows.
For $i\in N$, let
$B_i=A\setminus ((\bigcup_{\ell=t_i}^{k_i}E_i^{\ell}) \cup \bigcup_{j\in N\setminus \{i\}}\bigcup_{\ell=t_j+1}^{k_j} E_{j}^{\ell}))$. We check whether $|B_i|\geq k$ or not. If $|B_i|\geq k$, we know that there exists a subset of $B_i$, that is strictly more preferred by $i\in N$ and at least as preferred by each agent. The reason is that $B_i$ contains a more preferred worst alternative for agent $i$ than $D$ and contains at least as preferred worst alternative for other agents $j$ than $D$.
If $|B_i|<k$, then this means that a Pareto improvement with $i$ strictly improving is only possible if the size of the winning set is less than $k$ which is not feasible.
								\end{proof}
								
								We  now consider
								strategyproofness  together with ${\cal W}$-efficiency.
								We first note that Algorithm~\ref{algo:comtRSD} may not return a $\mathcal{W}$-efficient outcome. However, we construct a suitable strategyproof and $\mathcal{W}$-efficient by formalising an appropriate serial dictatorship algorithm for the worst set extension.

														\begin{theorem}
																					\label{th:sp-worst} There exists a linear-time and strategyproof algorithm that returns a   ${\cal W}$-efficient committee.
\end{theorem}
																 	\begin{proof}
												Consider the agents in a permutation $\pi$. The set of alternatives $A'$ is initialized to $A$. We reduce the set $A'$ while ensuring that it of size at least $k$. The next agent $i$ in the permutation comes and deletes the maximum number of least preferred equivalence classes from his preferences and the corresponding alternatives in $A'$ while ensuring that $|A'|\geq k$.
Each successive agent in the permutation gets a most preferred outcome while ensuring that agents before him in the permutation get at least as preferred an outcome as before. Thus the algorithm is strategyproof and Pareto optimal with respect to the `worst' set extension.
																\end{proof}
								
	\section{Upward Lexicographic Set Extension}

	We point out that for dichotomous preferences, the responsive set extension coincides with the 
	upward lexicographic  set extension.
	Hence we get a corollary of our results for responsive preferences:

	\begin{corollary}\label{cor:verify-ul}
							Checking whether a committee is  $UL$-efficient
							 is coNP-complete, even for dichotomous preferences and
			$Topwidth(\pref) \geq 3$.
	\end{corollary}

Note that Algorithm~\ref{algo:comtRSD} returns a  $UL$-efficient  committee. The reason is that each agent in her turn refines the set of possible outcomes to her most preferred subset of outcomes. Each committee is the refined set is at least as preferred with respect to RS (and hence with respect to UL) to all committees in the set of possible outcomes. 

		\begin{theorem}\label{th:ul-compute}
			There exists a linear-time and strategyproof algorithm that returns a  $UL$-efficient committee.
			\end{theorem}							

\section{Conclusions}

We considered Pareto optimality in multi-winner voting with respect to a number of prominent set extensions. We presented results on the relations between the notions as well as complexity of computing and verifying Pareto optimal outcomes.
Further directions of future work include considering Pareto optimality with respect to other set extensions~\citep{BrBr11a}. Another direction is consider the compatibility of Pareto optimality concepts with other axioms. 
Finally, we remark that our serial dictatorship algorithm can be used to define a multiwinner generalization of  random serial dictatorship,  which is worth investigating and raises interesting computational problems.

\section*{Acknowledgments}

This is the extended  version of the IJCAI conference paper~\citep{ALL16a}.
 The authors thank Felix Brandt for useful pointers and comments. They also thanks the reviewers and attendees of IJCAI 2016 and COMSOC 2016 for useful comments. 
 J\'er\^ome Lang and J\'er\^ome Monnot thank the ANR project CoCoRICo-CoDec. 

\bibliographystyle{plainnat}

%


\end{document}